\documentclass[journal]{IEEEtran}
\makeatletter
\long\def\@makecaption#1#2{\ifx\@captype\@IEEEtablestring%
\footnotesize\begin{center}{\normalfont\footnotesize #1}\\
{\normalfont\footnotesize\scshape #2}\end{center}%
\@IEEEtablecaptionsepspace
\else
\@IEEEfigurecaptionsepspace
\setbox\@tempboxa\hbox{\normalfont\footnotesize {#1.}~~ #2}%
\ifdim \wd\@tempboxa >\hsize%
\setbox\@tempboxa\hbox{\normalfont\footnotesize {#1.}~~ }%
\parbox[t]{\hsize}{\normalfont\footnotesize \noindent\unhbox\@tempboxa#2}%
\else
\hbox to\hsize{\normalfont\footnotesize\hfil\box\@tempboxa\hfil}\fi\fi}
\makeatother
\usepackage[utf8]{inputenc}
\usepackage[dvips]{graphicx}
\usepackage{amsmath}
\usepackage{amssymb}
\usepackage{amsthm}
\usepackage{amsfonts}
\usepackage{array}
\usepackage{bm}
\usepackage{enumitem}
\usepackage{cite}
\usepackage{hyperref}
\usepackage{hhline}
\usepackage[dvipsnames]{xcolor}
\usepackage{makecell}
\usepackage{pgfplots}
\usepackage{afterpage}
\usepackage{fancyhdr}
\usepackage{setspace}
\usepackage{algorithmicx}
\usepackage[ruled]{algorithm}
\usepackage{algpseudocode}
\usepackage[most]{tcolorbox}
\usepackage[T1]{fontenc}
\usepackage{amsmath}

\usepackage{pdfpages}

\hypersetup{
colorlinks,
citecolor=black,
filecolor=black,
linkcolor=black,
urlcolor=black
}

\newtheorem{theorem}{\bf Theorem}[section]

\newtheorem{definition}{\bf Definition}[section]
\newtheorem{observation}{\bf Observation}[section]
\newtheoremstyle{case}{}{}{}{}{\textbf{}}{\textbf{:}}{ }{}
\theoremstyle{case}

\newtheorem{example}{\bf Example}[section]

\label{Latex Commands}
\pgfplotsset{compat = 1.5}
\begin{document}

\title{ Fast Periodicity Estimation and Reconstruction of hidden components from noisy periodic signal}

\author{
    \IEEEauthorblockN{Bharadwaj Aryasomayajula\IEEEauthorrefmark{1}, Dibakar Sil\IEEEauthorrefmark{2}, Sarbani Palit\IEEEauthorrefmark{1}}\\
    \IEEEauthorblockA{\IEEEauthorrefmark{1}Indian Statistical Institute, Kolkata, India
    \\\{mtc1613, sarbanip\}@isical.ac.in}\\
    \IEEEauthorblockA{\IEEEauthorrefmark{2}National Institute of Technology, Durgapur, India
    \\ds.20150096@btech.nitdgp.ac.in}
}


\maketitle

\begin{abstract}
Periodicity estimation from an arbitrary length
noisy signal is computationally very costly. A recently developed Ramanujan Fat Dictionary is one of the ways to find the hidden components from an arbitrary length (non integral multiple of period) of the signal. This method suffers from high run time due to the lack of information about the period and effect of noise on the signal. We propose a new method that efficiently estimates the period of the signal and finding the hidden components thus becomes easy from it. 
Our method works well with significantly low SNR values and runs in $O(n)$ time complexity, $n$ being length of the signal. Comparision of run time analysis between our method  for period estimation of a given signal and SVD method at various SNR values has been made and the corresponding hidden components are there by extracted by projecting onto the factor-Ramanujan Subspaces.

\end{abstract}

\IEEEpeerreviewmaketitle
\section{Introduction}
\IEEEPARstart{M}{}ost of the real time signals acquired by various sources are generally corrupted by noise. Extraction of weaker periodic components in presence of noise is a well known challenge in the field of digital signal processing. A lot of available works have dealt with the component signal extraction from a composite signal while the components are sinusoidal. Although, in practical situations, in many a case, periodic signals that are non-sinusoidal are encountered. 
The period length, the structural patterns of periodic segments and the relative strengths of the respective segments are the three fundamental features of periodic signal.

In continuous domain, a signal $x(t)$ is said to be periodic with period $T$, if
\begin{align}
x (t+T) &= x (t), \forall t
\label{eq:contperiod}
\end{align}
where $T$ is the smallest positive interval that satisfies the Eq.\ref{eq:contperiod}. Similarly, in discrete domain, periodicity is defined as follows:
\begin{align}
x [n + P] = x [n] , \forall n \in Z
\label{eq:disperiod} 
\end{align}
where $P$ is the smallest positive integer that satisfies Eq.\ref{eq:disperiod}\\
Over the years, many works have been done to estimate the period of a given signal along with the hidden periodic components of the signal.

Two or more periodic signals, when added, give rise to a periodic signal whose period is the least common multiple (lcm) or a factor of lcm of the two components \cite{Ramanujan1}.
Applying Fourier Transform (FT) is one of the simplest way to detect the periods where the inverse of the fundamental component gives the period of the signal. However, for  composite signals (signals with more than one hidden period), some components may not be strong enough to be detected which indicates that  spectral analysis is not a reliable method for such signals.

Periodic decomposition \cite{1} provides an approach for determining the hidden periods from a composite signal. Setharas and Staley proposed an algorithm in which instead of choosing pre-determined basis, a data dependent non- orthogonal  basis is introduced \cite{2}.
Their proposed method projected the signal onto subspaces of different periodicities. Comparison of the projection energies was used for period estimation.      
P. P. Vaidyanathan \textit{et.al.} introduced the concept of the Ramanujan subspace ($S_q$) on the basis of Ramanujan sums \cite{Ramanujan1}. The Ramanujan sums $C_q(n)$ and its circular shifts form an integer basis of the subspace $S_q$. Periodicity estimation of signals is achieved by representing the signals in terms of linear combination of its components belonging to a number of the Ramanujan sub-spaces. 
In the companion paper \cite{Ramanujan2}, Vaidyanathan \textit{et.al.} showed that the Ramanujan$-$sum expansion method does not hold good for FIR signals. A linear combination of the Ramanujan sums, having number of terms equal to that of the signal length and a Ramanujan subspace based method were proposed to handle the FIR case. The Ramanujan Periodic Transform was also introduced, and a relatively easy calculation involving the scaled integer periodicity matrices yields the periodicities present in a signal from the non-zero energy projections.

However, the approach proposed in \cite{Farey} has a limitation that the input signal must have a length which is an integral multiple of the hidden periods that are present in the signal. For example, if two signals having periods 7 and 13 respectively are super imposed to form a new signal then  to apply the method proposed in \cite{Farey}, the length of the newly formed composite signal should be an integral multiple of $13 \times 7=91$ (i.e. 91,182... etc.,).
\begin{figure}[ht]
	\centering
	\includegraphics[scale = 0.3]{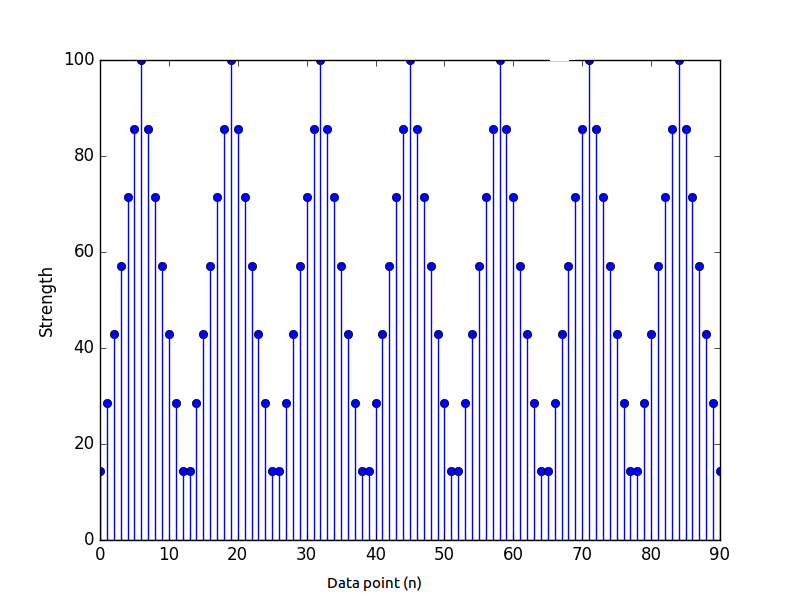}
	\caption{A triangular Signal with periodicity 13}
\end{figure}
\begin{figure}[ht]
	\centering
	\includegraphics[scale = 0.3]{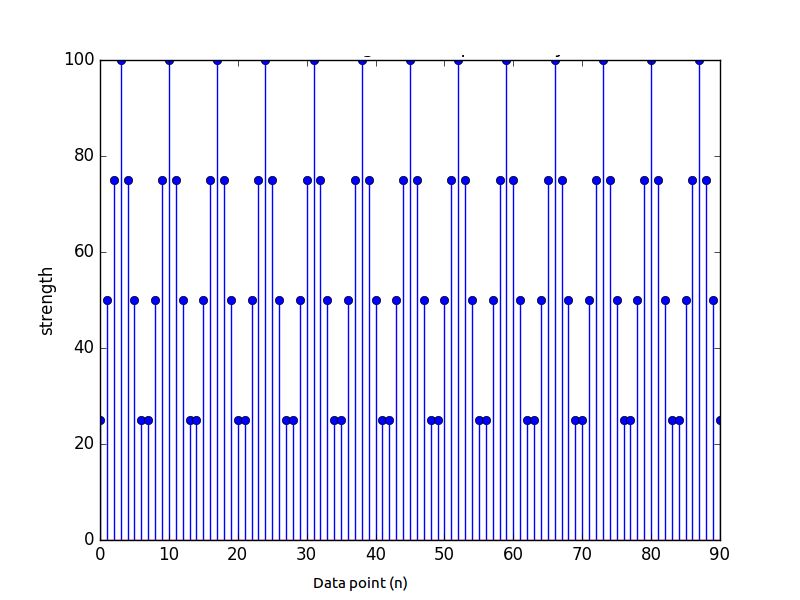}
	\caption{A triangular Signal with periodicity 7}
\end{figure}
\begin{figure}[ht]
	\centering
	\includegraphics[scale = 0.3]{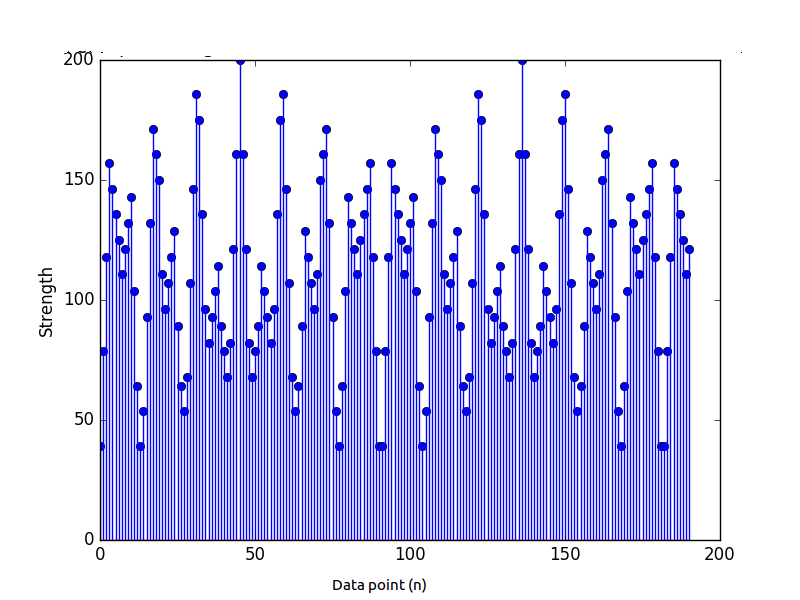}
	\caption{Signal with periodicity 91 when added Fig.1 and Fig.2}
\end{figure}
\begin{figure}[ht]
	\centering
	\includegraphics[scale = 0.3]{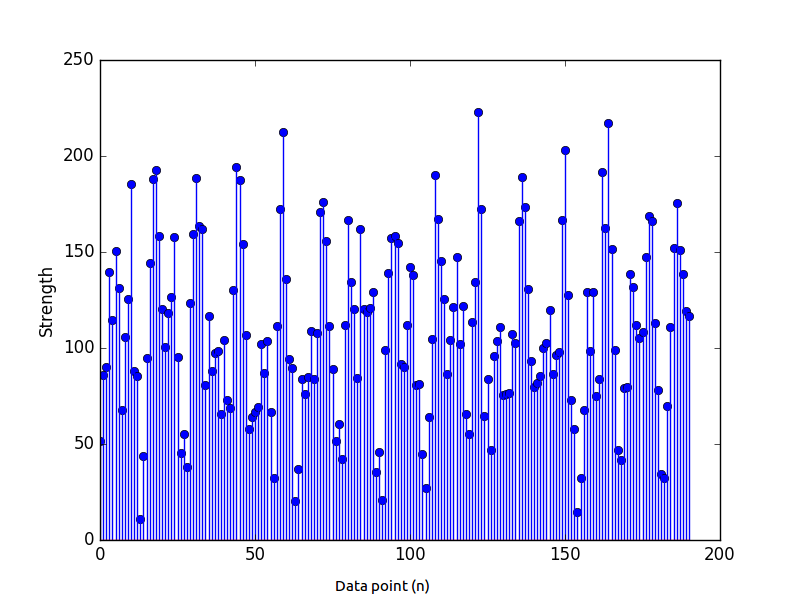}
	\caption{Noisy signal having SNR 35 dB}
\end{figure}

To overcome this problem, in \cite{Farey} dictionary or fat matrix based approach has been implemented for periodicity estimation of a signal. An integer valued generalization of the complex valued Farey dictionary has been also proposed. RPT based dictionaries proved to yield better result in terms of noise immunity and complexity.
The Farey Dictionary method in \cite{Farey} was further generalized in \cite{Nested} where a number of other dictionaries have been constructed which provide faster calculations. Inspired from the RPT matrices a family of square matrices have been introduced that results in faster calculation. Some real valued alternatives of the DFT matrix is used to form the Nested Periodic Matrices fusing the concepts of many transforms that are used in periodicity estimation. Incorporation of the $L_2$ norm based methods further simplifies the calculation. In \cite{Rfilter Prop} Ramanujan filter banks were introduced based on the Ramanujan-sums. In order to estimate periodicities with the help of filters, a filter must be designed which indicates presence of particular period in its input signal by a significant change in its output. Non adaptive comb filters based on the Ramanujan-sums serves the purpose. The Ramanujan filter banks also proved to be successful in representing the time varying periodicity of a signal.
The advantages of period detection based on RFB over Short-Time Fourier Transform (STFT) is described in \cite{Rfilterbank}. Further researches suggested that the Nested Periodic Dictionaries proposed in \cite{Rfilter Prop} are minimal in size \cite{Mindictionary}. For proper integer period estimation, the highest lower boundary of the data length that is subjected to an algorithm was aimed to obtain in \cite{Mindatalen}. The identifiability of the hidden periods has also been evaluated in terms of the input data length. After evolving the bound in data length that is applicable to any of the period estimation techniques, the results were further extended to the Farey dictionary based single and hidden period detection method. It was also shown that depending on the number of expected harmonics in a signal the required data length changes for signals with non integer periods.

\subsection{Our Contribution}
When the length of the signal ($N$) is not an integral multiple of the period of the signal then fat matrix technique presented in \cite{Farey} tries to figure out the period of the signal by projecting the signal on to all of the ramanujan subspaces $S_q$'s, $q < N$. This is rather time consuming. In such cases, period of the signal has to be estimated. This estimation of period of the signal when the length of the signal is not an integral multiple of the period is dealt here. We proposed a method to estimate the period of the signal given the above restriction and compared it with the Singular Value Decomposition (SVD) method described in \cite{SVD}. The proposed method has been made fast by employing randomization and montecarlo techniques.

The reconstruction of the periodic components from the composite signal corrupted by noise is also a great challenge which is dealt with in this studies. In some prior works, although the envelop of the periodic components are estimated with quite appreciable accuracy, the DC parts of the reconstructed components are not identified. In this study it is mathematically shown that the correlation between the original component signal and the reconstructed signal is maximum when the DC part of the whole signal is distributed equally among all of its constituents. The DC part of the whole signal is easily evaluated, as it is same as the projection of the whole signal to the ramanujan subspace $S_1$.

\section{Period Estimation}
Let us suppose we have a signal $X$ of length $N$ that is made up of $k$ hidden components $X_{p_1},X_{p_2},..X{p_k}$ (where each component $X_{p_i}$ ($0 < i < k $) is periodic with period $p_i$) is corrupted with some noise. Let the period of the composite signal be $p$.

Now, when $N \neq c \times p$ for some integer $c$ then we need to apply the fat matrix \cite{Farey} technique to figure out the hidden components. This involves in projection of the signal onto all the ramanujan spaces of size less than $N$ which is rather time consuming. So, we have to estimate the period of the signal and then we need to project the signal only onto the ramanujan spaces of sizes which are factors to the estimated period.

We first discuss a little about the SVD method for period estimation and then discuss our technique for the estimation of the composite period.
  
In the SVD technique for period estimation we make a data matrix of size $\lfloor \frac{N}{P} \rfloor \times P$ for an assumed period of $P < N$ . Period is estimated based on the $\lambda _1 / \lambda _2$ values obtained through the SVD process for the data matrices obtained for each assumed period $P$. When the assumed period $P$ is the actual period $p$ then we get maximum value for $\lambda _1 / \lambda _2$ and $p$ is identified.

\subsection*{Data Matrix formation:}
\label{DataMatrix}
Data Matrix $D_P$ for an assumed period $P$, is formed by dividing the signal into blocks of size $P$ and omitting the last portion of the data if its size is less than P  and then taking each block as a row of a $\lfloor \frac{N}{P} \rfloor \times P$ size matrix.
\begin{example}

For $X = [1,2,3,1,2,3,1,2]$ and for an assumed period of $P$ = 3, the data matrix $D_3$ is as follows:
\[
D_3  =
\begin{bmatrix}
1 & 2 & 3\\
1 & 2 & 3
\end{bmatrix}
\] 
\end{example}
\subsection{The Technique}
\label{periodfinder}
Instead of running SVD on the obtained data matrices we employ a different technique as presented below for the estimation of the period. It is also shown later in \ref{hiddenperiods} that information about the hidden periods as well can be obtained through this technique.
\begin{theorem}{\textit{Minimum Variance}.}\label{Minimum Variance}\\
The variance of each and every column vector of data matrix formed by chopping $X$ at blocks of $P$ is 0 iff $P$ is the composite period $p$ or a multiple of it.
\end{theorem}
\begin{proof}
\textbf{If Part}:\\
When the assumed period P is the composite period $p$ or a multiple of $p$ say $m\times p$ then data matrix $D_{mp}$ obtained is as follows:
\begin{align}
D_{mp} = 
\begin{bmatrix}
X[0] &\dots & X[i] &\dots &X[mp-1]\\
X[mp] &\dots & X[mp+i] &\dots & X[2mp-1]\\
\vdots &\dots &\vdots &\dots &\vdots\\
\end{bmatrix}
\end{align}
Now, consider any column of $D_{mp}$, say the $i^{th}$ column, $C_i$. We have,
\begin{align}
\label{eq:minvecvar}
C_i =
\begin{bmatrix}
X[i] & X[mp+i] & X[2mp+i] &\dots&
\end{bmatrix}
^T
\end{align}
As $p$ is the composite period of the signal, we have $X[np+i] = X[i]$ for some integer $n$. Thus, Eq.\ref{eq:minvecvar} boils down to 
\begin{align}
C_i =
\begin{bmatrix}
X[i] & X[i] & X[i] &\dots&
\end{bmatrix}
^T
\end{align}
As all the elements of the columns are same, the variance is zero.\\
\textbf{Only If Part}:\\
Let us suppose there exists a $P'$ for which all the columns of the data matrix ($D_P$) formed by chopping $X$ at blocks of $P'$ is zero and $P'$ is not a multiple of $p$.

Since, each and every column of $D_{P'}$ has variance zero we have that all the elements of any given column to be same. This is possible only if $P'$ is a period or multiple of a period which is a contradiction.     
\end{proof}

Our idea is to form data matrix for each assumed period to the noisy signal $P$ and find variance of all columns. The value of $P$ for which the variance is minimized is our period.

\subsection{Hidden Periods}
\label{hiddenperiods}
The method described in \ref{periodfinder} will tell us about the hidden periods as well.

\begin{theorem}{Hidden Components\\}
\label{thm:hidden period}
Let $X_{p_i}$ be one of the component of $X$ that has a period of $p_i$. Now, when the data matrix $D_{mp_i}$ formed by chopping at lengths of $mp_i$ is considered the variance of each column has contributions due to the other component signals only.
\end{theorem}
\begin{proof}
We have data matrix $D_{mp_i}$ as:
\begin{align}
D_{mp_i} & = 
\begin{bmatrix}
X[0] &\dots & X[i] &\dots &X[mp_i-1]\\
X[mp_i] &\dots & X[mp_i+i] &\dots & X[2mp_i-1]\\
\vdots &\dots &\vdots &\dots &\vdots\\
\end{bmatrix}
\end{align}
Now, consider any column of $D_{mp_i}$, say the $j^{th}$ column, $C_j$. Let us denote $X' = X - X_{mp_i}$. We have,
\begin{align}
\begin{split}
C_j
& =
\begin{bmatrix}
X[j] & X[mp_i+j] & X[2mp_i+j] &\dots \\
\end{bmatrix}^T \\
& = 
\begin{bmatrix}
X_{mp_1}[j] & X_{mp_1}[mp_i+j] & X_{mp_1}[2mp_i+j] \dots
\end{bmatrix} ^T \\  
& \quad + 
\begin{bmatrix}
X'[j] & X'[mp_i+j] & X'[2mp_i+j] \dots \\
\end{bmatrix} ^T\\
& =
\begin{bmatrix}
X_{mp_i}[j] & X_{mp_i}[mp_i+j] & X_{mp_i}[2mp_i+j] \dots
\end{bmatrix} ^T \\
& \quad + 
\begin{bmatrix}
X'[j] & X'[mp_i+j] & X'[2mp_i+j] \dots \\
\end{bmatrix} ^T \\
& =
\begin{bmatrix}
X_{mp_i}[j] & X_{mp_i}[j] & X_{mp_i}[j] \dots
\end{bmatrix} ^T \\
& \quad + 
\begin{bmatrix}
X'[j] & X'[mp_i+j] & X'[2mp_i+j] \dots \\
\end{bmatrix} ^T \\
\end{split}
\end{align}
Variance of a data set is invariant to constant addition. This concludes the proof. 
\end{proof}

Let us call the plot of cumulative variance of all the columns against assumed periods as variance graph or variance plot.
\begin{figure}[ht]
	\centering
	\includegraphics[scale = 0.4]{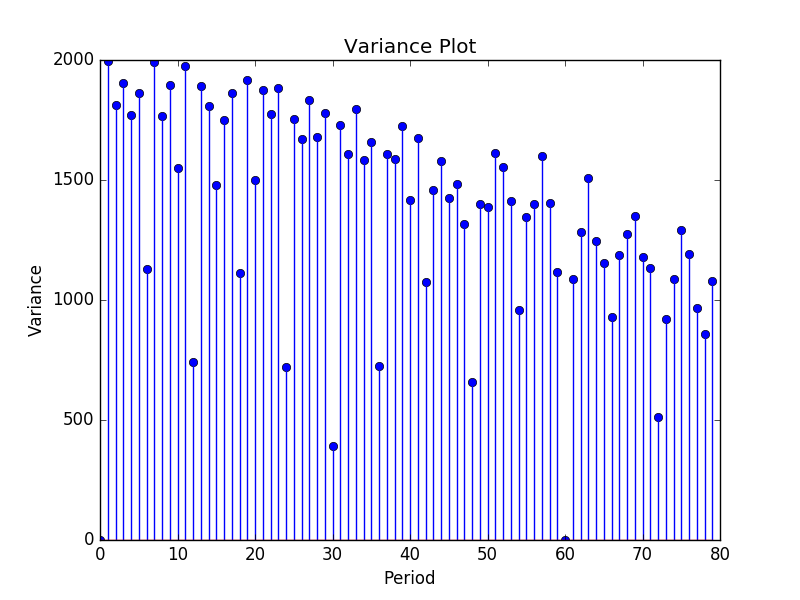}
    \caption{Variance Plot of the signal when there is no noise in the signal}
\end{figure}
\begin{figure}[ht]
	\centering
	\includegraphics[scale = 0.3]{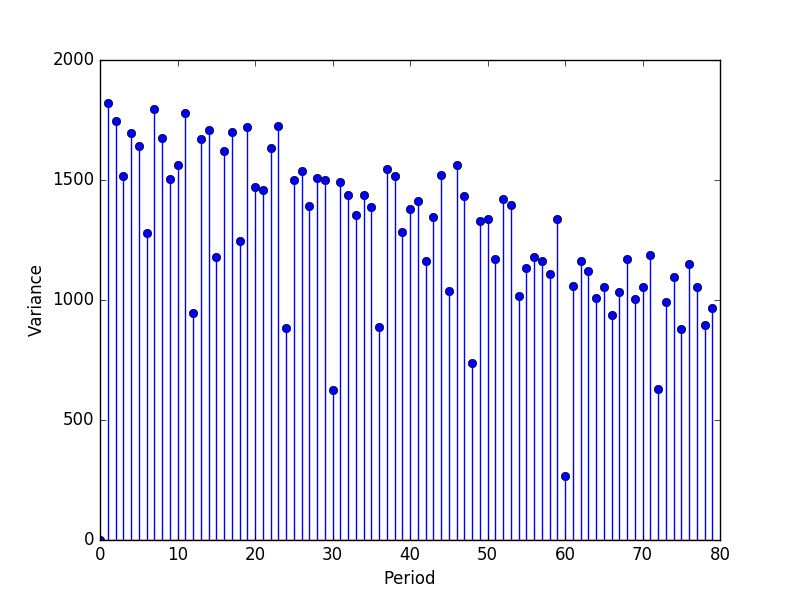}
    \caption{Variance plot of the signal when it is corrupted with 14.56 dB of SNR}
\end{figure}
\begin{figure}[ht]
	\centering
	\includegraphics[scale = 0.3]{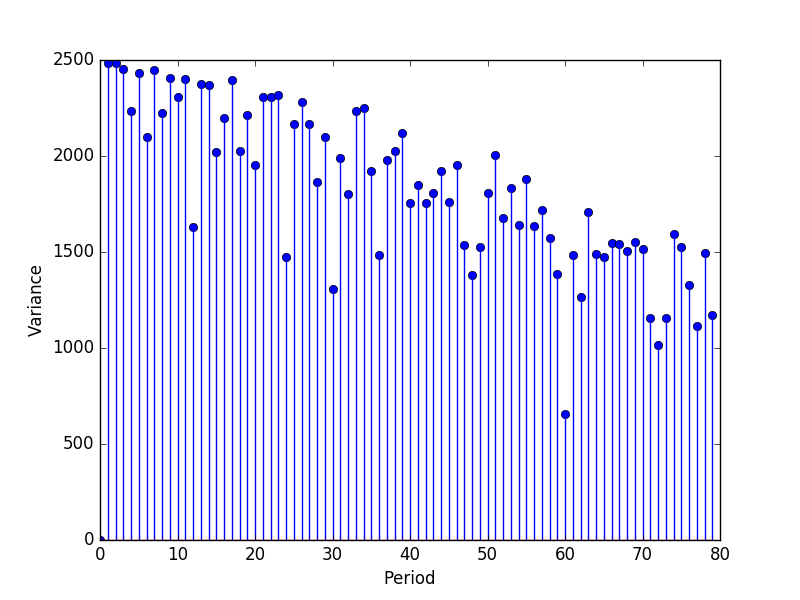}
    \caption{Variance plot of the signal when it is corrupted with 9 dB of SNR}
\end{figure}

\begin{definition}{Dip and Dip Magnitude\\}
We define dip as the point where the magnitude of variance is less than its immediate neighbours in the variance graph.

If there is a dip at an assumed period $P$, in the variance graph,\textit{var}, we defined the following dip magnitude measures on it:
\begin{align}
\begin{split}
1 &\quad max(var) - var[P]\\ 
2 &\quad max(var) + var[P-1] + var[P+1] - 3var[P]\\ \label{dip-mag-2}
\end{split}
\end{align}  
\end{definition}

\begin{observation}
Dips can be observed at assumed period lengths of integral multiples of hidden periods in a variance graph. 
\end{observation}
\begin{proof}
A 'dip' in variance is observed in the variance graph when data matrices are constructed by chopping the signal at block lengths of hidden periods or their multiples; as contribution to the variances due to corresponding signal component is zero by theorem \ref{thm:hidden period}
\end{proof}

As the composite period is the (least common) multiple of each and every hidden period a greater dip is expected at the composite period.

\subsection{The Algorithm}
In this section we discuss about two algorithms that we implemented based on the Theorems \ref{Minimum Variance} and \ref{thm:hidden period}.
\subsubsection{Period Finder Algorithm}\label{algo:period finder}
First we create the data matrices $D_P$'s for various assumed periods $P$ and variance graph is obtained. This variance graph is processed to identify dips and from here a dip graph is obtained. A dip graph is the plot of dip-magnitude vs assumed period $P$ where ever a dip is found. (Here dip-magnitude is taken to be the second measure defined on the dip.)

As per the theorems \ref{Minimum Variance} and \ref{thm:hidden period} a dip corresponds to a multiple of hidden period or a hidden period itself. The value at which dip-graph has highest magnitude should correspond to the composite period $p$ provided the noise is not periodic.

Estimation of the composite period depends on the strength of the noise present in the signal. If the noise is too high then it corrupts the value of the variance at the composite period so much that the estimated period may not be the composite period.

\begin{observation}
The dip-magnitude at a given hidden period or a multiple of it depends on the relative strength of the hidden component w.r.t. the remaining other components apart from the noise present to the signal
\end{observation}
\begin{proof}
From Theorem\ref{thm:hidden period} we have seen that magnitude of variance of the data matrix $D_P$ formed by chopping $X$ at lengths of $P$ which is a hidden period will have no contribution due to the component $X_P$. Stated other wise, has contribution because of the other components only. Thus, if the relative strength of $X_P$ w.r.t to the remaining components as well as noise is less then it will have only a small dip-magnitude at that corresponding hidden period.
\end{proof}
\alglanguage{pseudocode}
\begin{algorithm}[H]
\small
\caption{Minimum Variance Period Finder}\label{alg:MVPF}
\begin{algorithmic}[1] 
\State var = zeroes($\lfloor N/2 \rfloor$)
\State dip = zeroes($\lfloor N/2 \rfloor$)
\Procedure{Period Finder}{$\overset{\wedge}{X}$}
	\For {j in 2 to $\lfloor N/2 \rfloor$}
		\State Create data matrix $D_j$
			\For {i in 1 to j}
				\State var[j-1] = var[j-1] + variance($D_j$[:,i])
			\EndFor
			\State var[j-1] = var[j-1]/j
	 \EndFor
	 \State prev = 1; cur = 2, nxt = 3;
	 \For {j in 1 : $\lfloor N/2 \rfloor$ - 1}
	 	\If {var[cur] < var[prev] \textbf{and} var[cur] < var[next] } 
	 		\State dip[cur] = maximum(var) - var[cur]
	 		\State dip[cur] += var[prev]+var[next] - 2*var[cur]
	 		\State dip[cur] = (dip[cur])$^4$
		\EndIf
	\EndFor
\State \Return maximum(dip)
\EndProcedure 
\end{algorithmic}
\end{algorithm}
Step (5) of the algorithm creates a data matrix $D_j$ at an assumed period $j$. Steps (6-9) will calculate the cumulative variance of all columns of $D_j$ averaged over $j$. These steps (5-9) run in $O(n^2$) time and step(4) loops over this $O(n^2)$ process. Thus steps (4-10) run in $O(n^3$) time. Steps (11-19) extract the period with maximum dip magnitude. Here, the dip magnitude is taken as the fourth power of second dip-magnitude measure we defined. This extraction runs in $O(n)$. Thus entire procedure runs in $O(n^3$) time.
\subsubsection{Modified period finder - Monte Carlo Algorithm}
We have another important observation to make here before going further
\begin{observation} \label{obs:subset}
The variance graph obtained by taking variance of a subset of every column vector of some constant length taken from the data matrix formed by chopping $X$ at various assumed period lengths would estimate the composite period correctly.  
\end{observation}
\begin{proof}
At the composite period irrespective of the length of the subset of the column vector considered variance always depends on the noise. With in permissible levels of the noise, always the value of variance obtained at the composite period is minimum.
\end{proof}
To improve upon the run time bound of the \ref{algo:period finder}, instead of accumulating and averaging over all the columns of the data matrix we consider only few (constant) number of columns. Also, by observation\ref{obs:subset} even in these columns, a subset of them is sufficient. So we choose only a fixed length subset of columns from this data. So effectively the algorithm runs in $O(n)$
To counter the effect of error that may enter due to selecting only few number of columns, we will process this using a Monte-Carlo type approach by asking the user to send the signal few($k$) many times and only those assumed periods that are consistently present in all the k runs are considered for period estimation task.

\alglanguage{pseudocode}
\begin{algorithm}[H]
\small
\caption{Monte Carlo Period Finder}\label{alg:MCPF}
\begin{algorithmic}[1] 
\State var = zeroes($\lfloor N/2 \rfloor$)
\State dip = zeroes($\lfloor N/2 \rfloor$)
\State count = zeroes($\lfloor N/2 \rfloor$)
\State resends = $k$
\Procedure{Monte Carlo}{$\overset{\wedge}{X}$}
	\For {l in 1 to resends}
		\For {j in 2 to $\lfloor N/2 \rfloor$}
			\State Create data matrix $D_j$
				\For {i in 1 to j}
					\State var[j-1] = var[j-1] + variance($D_j$[:,i])
				\EndFor
				\State var[j-1] = var[j-1]/j
	 	\EndFor
	 	\State prev = 1; current = 2, next = 3;
	 	\For {j in 1 : $\lfloor N/2 \rfloor$ - 1}
	 		\If {var[current] < var[prev] \textbf{and} var[current] < var[next] } 
	 			\State dip[current] = maximum(var) - var[current]
	 			\State dip[current] += var[prev]+var[next] - 3var[current]
	 			\State dip[current] = (dip[current])$^4$
			\EndIf
		\EndFor
	\State \Return maximum(dip)
	\EndFor
\EndProcedure 
\end{algorithmic}
\end{algorithm}

\section{Finding the hidden components} \label{hidden components}
Let us suppose the hidden components in the given signal $X ^{'}$ be $X_{p_1},X_{p_2},X_{p_3},....,X_{p_k}$. Let us also denote the corrupted version of $X^{'}$ by $\overset{\wedge}{X^{'}}$ Now using the process discussed as in \cite{Ramanujan1}, \cite{Ramanujan2} we project $\overset{\wedge}{X^{'}}$ on to various ramanujan sub spaces that are factors of the estimated composite period $p$. Then $\{p_1,p_2,p_3,...,p_k \}$ are factors of $p$ where $p_1 = 1$ and $p_k = p$ are trivial factors. As $p_1 = 1$ we have corresponding $X_{p_1} = X_1$ to be periodic with period 1. Hence $X_1$ corresponds to just a dc signal. Let the dc level be $d$. We know that ramanujan spaces are mutually orthogonal i.e.,
\begin{equation}
\label{dot product}
X_i^\dagger X_j = \delta _{ij}
\end{equation}   
Now, taking $i = 1$ in Eq.\ref{dot product} and j $\neq$ 1, we have,  
\begin{align}
& \quad\quad  X_1^\dagger X_j = \delta _{1j} \\
& \Rightarrow d \sum\limits_{n = 1}^{n = p}X_j[n] = 0 \\
& \Rightarrow   \sum\limits_{n = 1}^{n = p}X_j[n] = 0 \\
\end{align}
i.e., all $X_j$'s with j $\neq$ 1 have dc value 0. So, the dc component present in each hidden component is transferred to the $X_1$ and all the rest of hidden components are zero mean components. Hence,  hidden components obtained are shifted versions of the actual hidden components (the shift is so much that their dc value is zero) and the dc signal $X_1$ has dc value ($d$) such that
\begin{align}
dc(X_{reconstructed}) = dc(X_1) = dc(X^{'})
\end{align}

When we have no information about the dc value of each hidden component then the reconstructed components are shifted versions of actual hidden components. Let us denote the actual hidden compnents by $X_{i_{act}}$. Then we have 
\begin{align}
X_i + \alpha _i X_1 = X_{i_{act}} \quad \forall i \neq 1 
\end{align}
such that
\begin{align}
\sum \alpha _i = 1 
\end{align}

When information about dc values of hidden components is not available then the $\alpha _i$'s cannot be determined.
Let us suppose we have arbitrarily assigned some value ($\alpha _i ^{'}$) to each $\alpha _i$ in this case. i.e.,
\begin{align}
X_i + \alpha _i ^{'} X_1 = X_{i} ^{'} \quad \forall i \neq 1
\end{align}
and let X represent the sum of all the estimated hidden components. That is,
\begin{align}
\sum\limits_{p_2}^{k} X_{i} ^{'} = X 
\end{align}
We show that all the $\alpha _i$'s should be chosen to be same for the correlation between $X'$ and $X$ to be maximum when no information on dc values of hidden. components is provided.

This can be formulated as follows:\\
\textbf{Maximize}: $X^\dagger X^{'}$\\
\textbf{Subjected to}:
\begin{align}
& X_i + \alpha {_i ^{'}} X_1 = X_{i} ^{'} \quad \forall i \neq 1\\
& \sum \alpha {_i ^{'}} = 1 
\end{align}
This optimization problem is solved using Lagrange method of undetermined multipliers and it yields the result that all $\alpha _i$'s should be same.  

Thus when we have no information about the dc value of the hidden components then to maximize the correlation we need to distribute $X_1$ equally to the rest of the components constructed.
\section{Case Studies}
Here we show various results obtained using the period estimation technique and comparison between our proposed method and SVD method for period estimation as well as signal estimation based on the estimated period are shown. 

We consider the hidden periods of a composite signal are $p_1$ = 8, $p_2$ = 11, $p_3$ = 16. A random signal is added to the composite signal and finally it is considered to be a noisy signal. Let the data length is 4119 whereas the composite period of the signal is 176.

\begin{figure}[H]
	\centering
	\includegraphics[scale = 0.3]{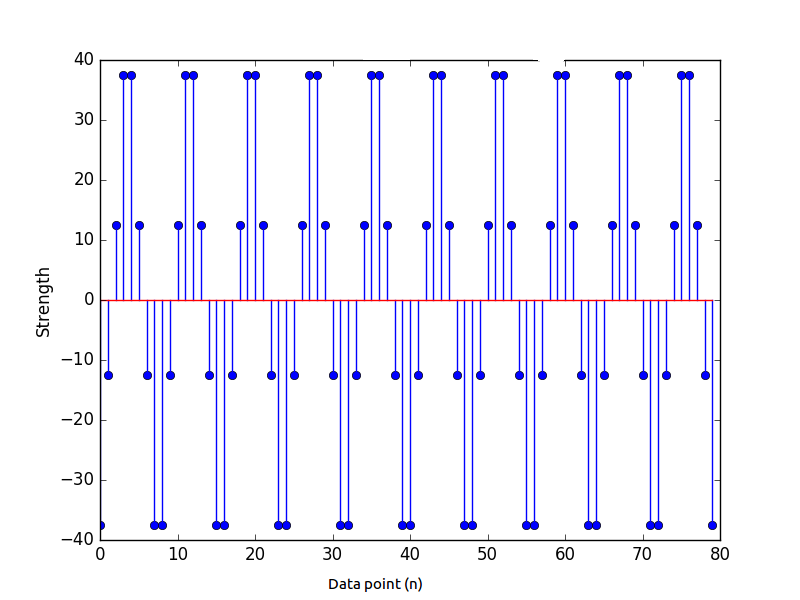}
    \caption{Mean Zero triangular periodic signal with period 8}
\end{figure}
\begin{figure}[H]
	\centering
	\includegraphics[scale = 0.3]{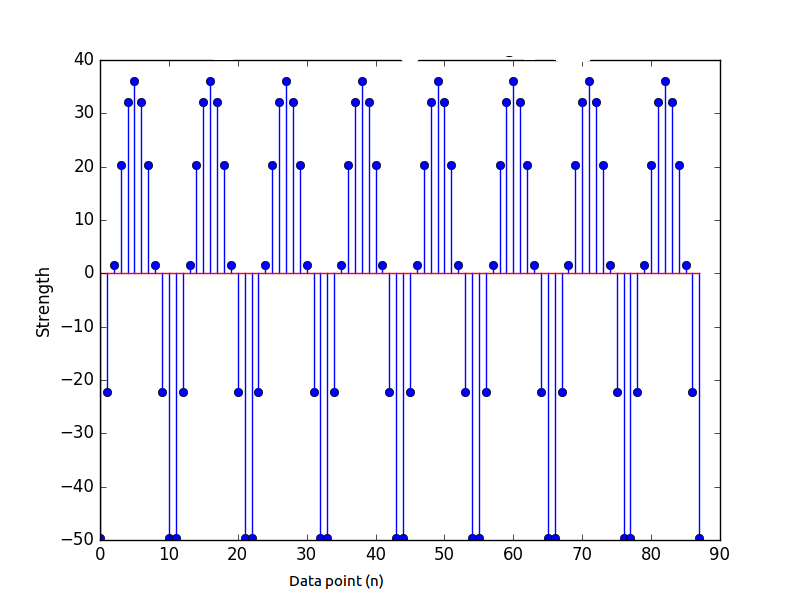}
    \caption{Mean Zero cosine signal with period 11}
\end{figure}
\begin{figure}[H]
	\centering
	\includegraphics[scale = 0.3]{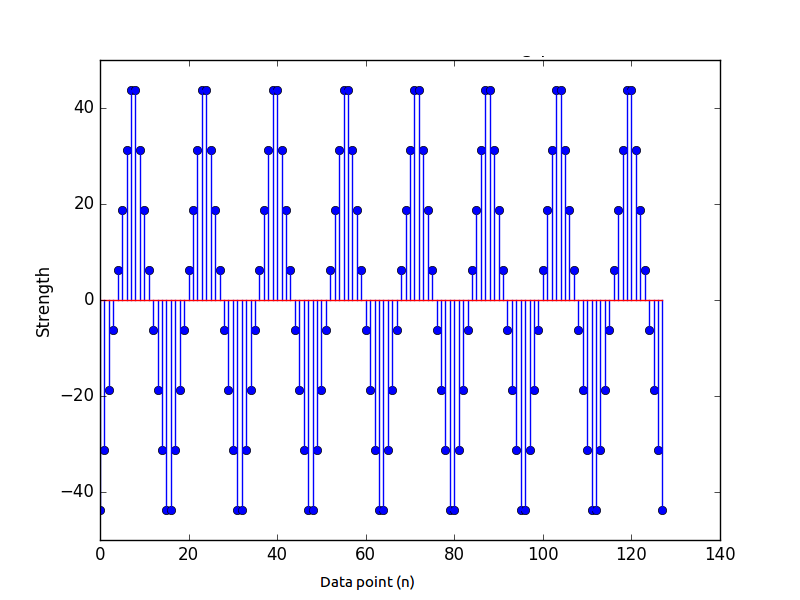}
    \caption{Mean Zero triangular periodic signal with period 16}
\end{figure}
\begin{figure}[H]
	\centering
	\includegraphics[scale = 0.3]{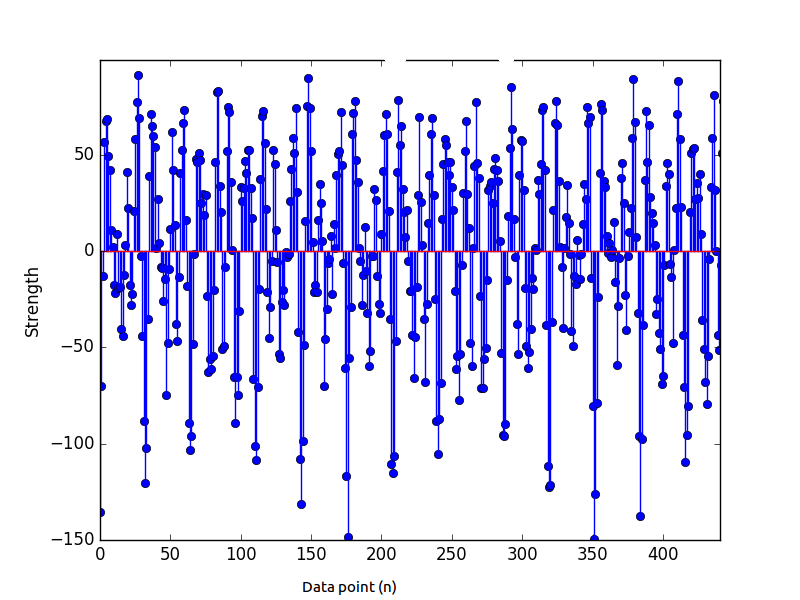}
    \caption{This is the noisy signal composed of above hidden periodicity and a mean zero gaussian noise of 32 dB SNR }
\end{figure} 
From the above noisy signal, to find the composite periods, below is two processes:
\begin{figure}[H]
	\centering
	\includegraphics[scale = 0.3]{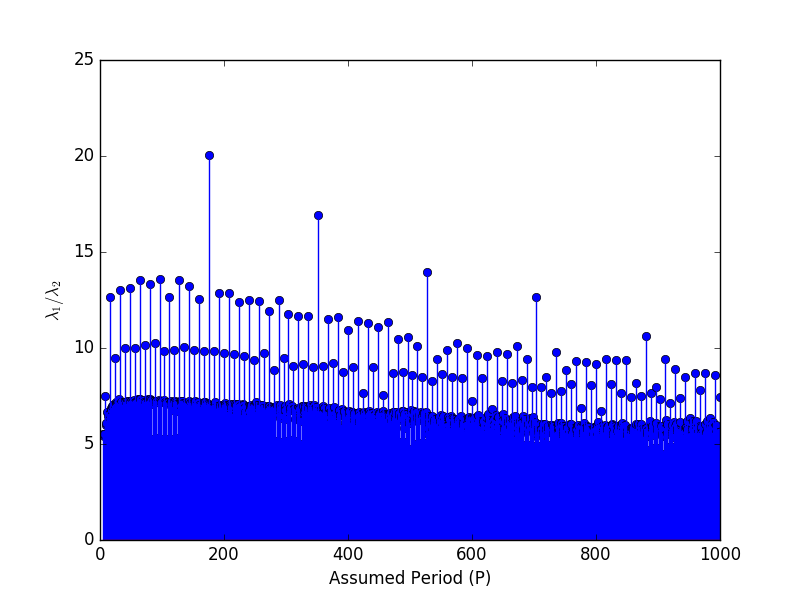}
    \caption{SVD Spectrum of the signal}
\end{figure}
\begin{figure}[H]
	\centering
	\includegraphics[scale = 0.3]{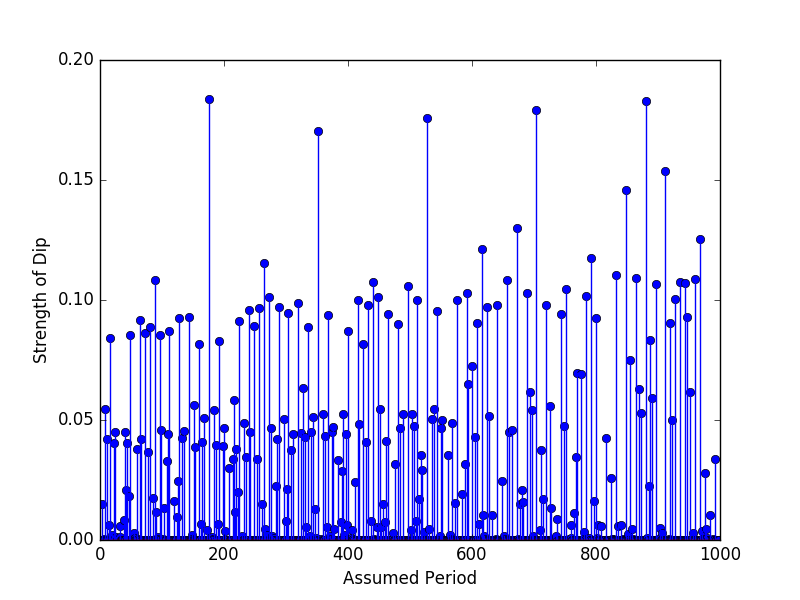}
    \caption{Variance vs. Assumed period}
\end{figure}
In this figure, only hidden or multiple of hidden and composite or multiple of composite period is present rest of all the points are zero. \\
Another advantage of using our proposed method is to compute the composite period length or an integral multiple of composite period length from a very long data length, is very fast, i.e. an order of $O(n)$ algorithm is proposed in this study.
\begin{figure}[H]
	\centering
	\includegraphics[scale = 0.4]{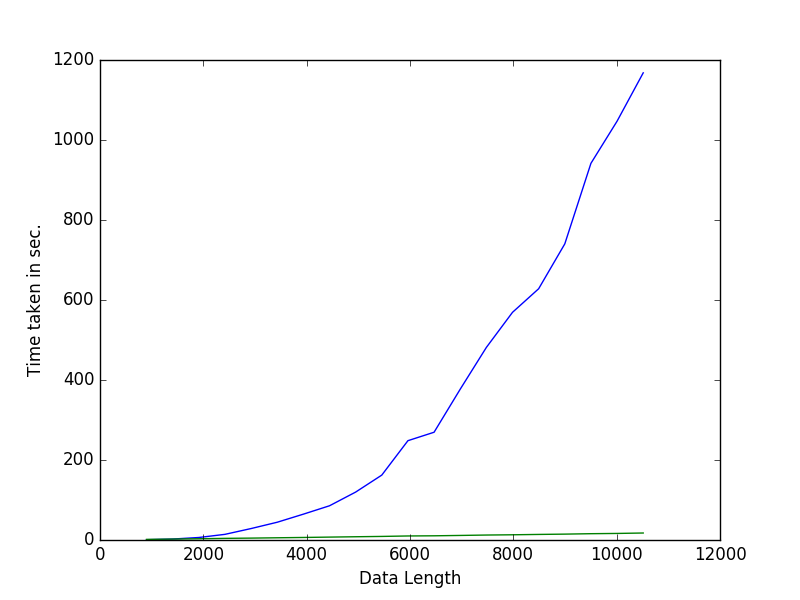}
    \caption{SVD Spectrum vs. Our proposed methodology run time analysis}
\end{figure}
Here the green line is the plot of our proposed method time taken curve and the blue one is the $SVD$ method time taken curve.
\begin{figure}[H]
	\centering
	\includegraphics[scale = 0.3]{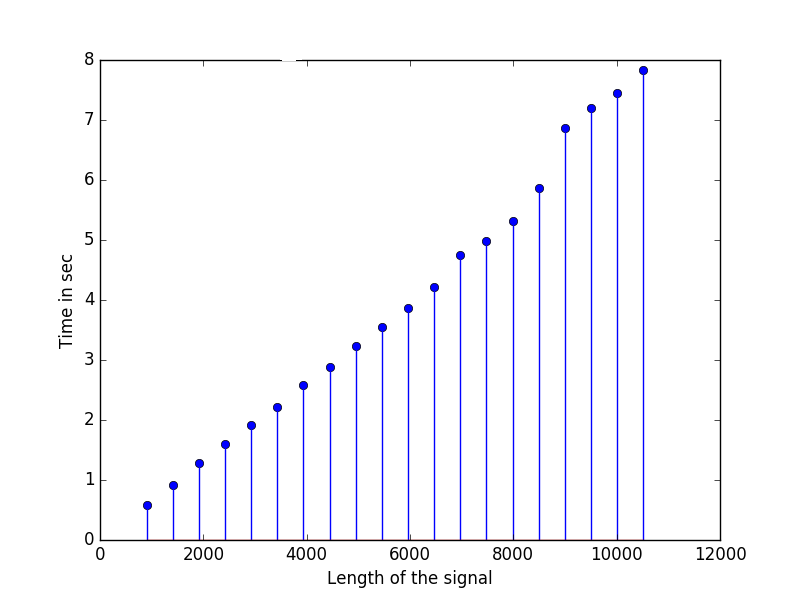}
    \caption{Montecarlo Simulation}
\end{figure}

\begin{figure}[H]
	\centering
	\includegraphics[scale = 0.3]{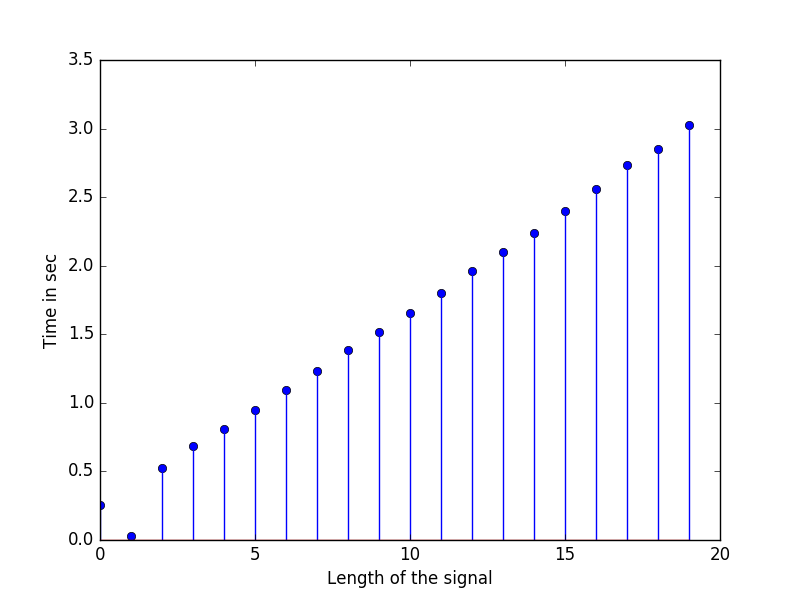}
    \caption{Montecarlo Simulation of time taken by the our proposed methodology when resend value is 2}
\end{figure}

This is the curve of Hit-Miss of finding actual composite period. It shows out of 20 iterations, only $2nd$ time it failed when resends is 2. So, here the hit-miss ratio is $19:1$
\begin{figure}[H]
	\centering
	\includegraphics[scale = 0.3]{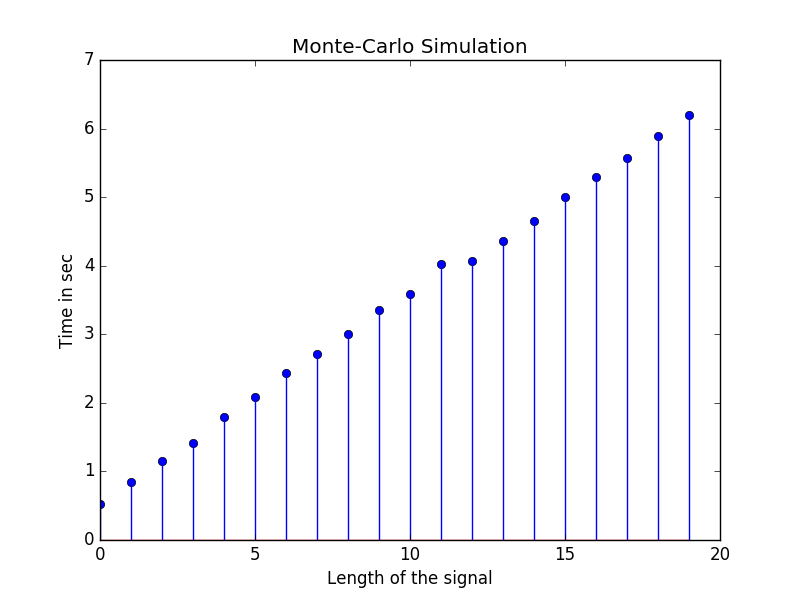}
    \caption{Montecarlo Simulation of time taken by the our proposed methodology when resend value is 5}
\end{figure}

This is the Hit-Miss curve for resends value 5. It shows there is no miss of the composite period detection. So, here the hit-miss ratio is $20:0$ (as there is no miss).
\begin{figure}[H]
	\centering
	\includegraphics[scale = 0.3]{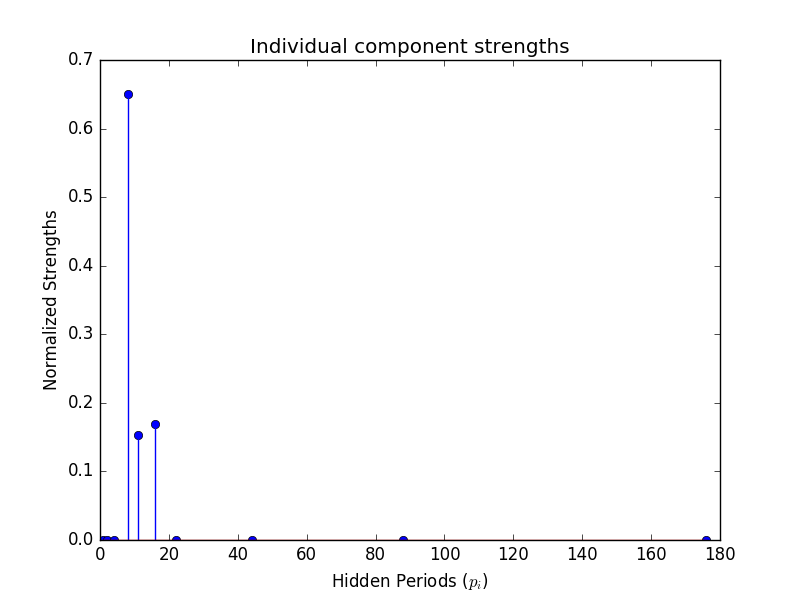}
    \caption{Individual periodic component's normalized strength at an SNR of 35 dB}
\end{figure}

This is the whole signal gets projected on the factor ramanujan space of composite period length. Which also shows that most of the energy contributed by the individual hidden components itself. 
\begin{figure}[H]
	\centering
	\includegraphics[scale = 0.3]{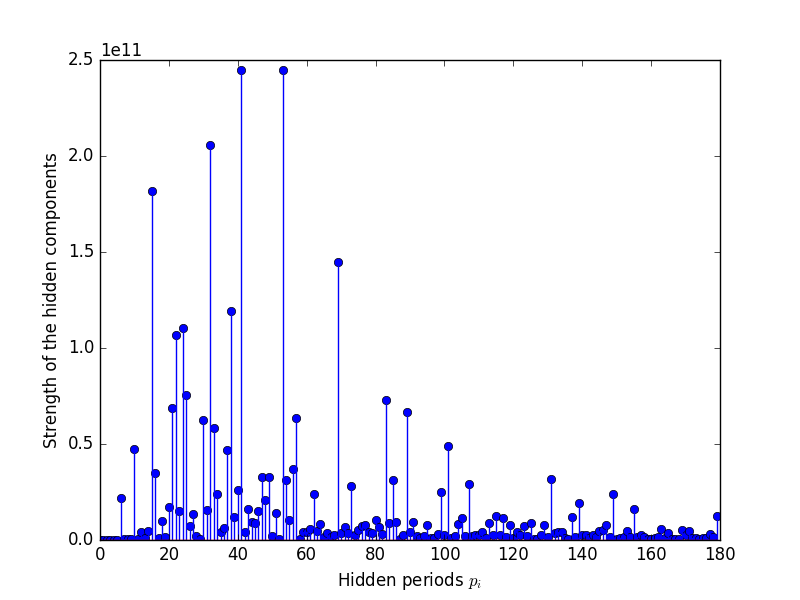}
    \caption{Individual periodic component's strength at a SNR of 35 dB}
\end{figure}
This the Fat-Matrix solution of the whole signal. \\
So, from the fat matrix solution, it is impossible to comment about the hidden periods. Whereas, if we project the signal truncating it at the composite period length or the multiple of composite period length, it is easy to identify the actual hidden components present in a signal.  \\
After following our process, we truncated the signal at the multiple of composite periods and then projected this signal at the factor ramanujan subspace of that composite period to find the hidden components. Below is the figure of the highest error valued component strengths w.r.t the SNR values.
\begin{figure}[H]
	\centering
	\includegraphics[scale = 0.3]{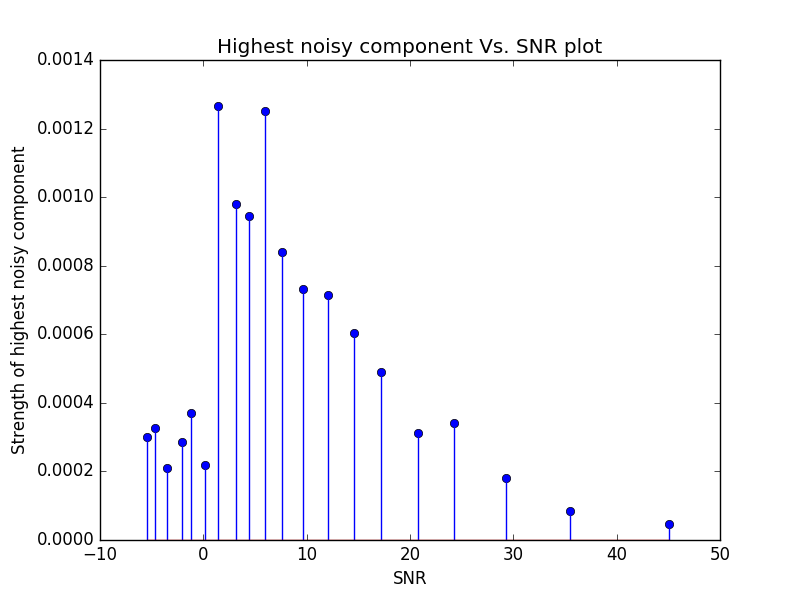}
    \caption{Highest noisy component vs. SNR plot}
\end{figure}
This is the plot of $4th$ highest component normalized strength vs. SNR (in dB) is plotted. Which shows the highest noisy component strength is very low that it won't effect the actual signal. \\
After getting the information about the hidden components, reconstructed hidden periodic signal forms are found. Those are below:
\begin{figure}[H]
	\centering
	\includegraphics[scale = 0.3]{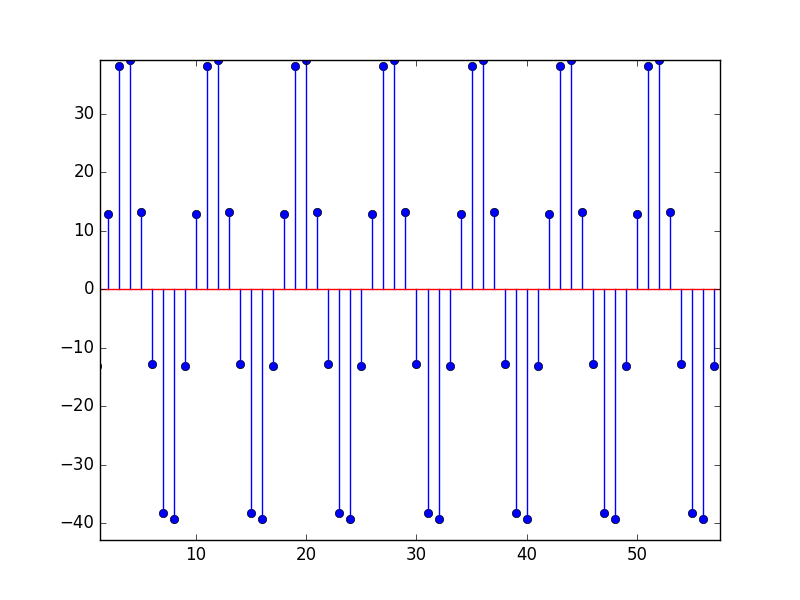}
    \caption{Reconstruction of periodic component 7}
\end{figure}

This is the reconstructed $1st$ periodic component $p_1 = 8$. From the noisy signal of noise margin, SNR of 7.39 dB.
\begin{figure}[H]
	\centering
	\includegraphics[scale = 0.3]{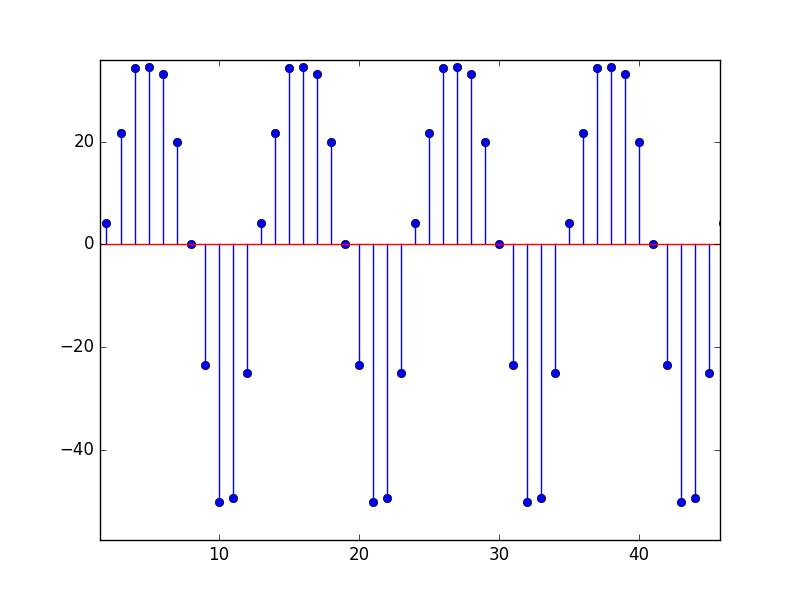}
	\caption{Reconstruction of periodic component 11}
\end{figure}

This is the $2nd$ reconstructed periodic component $p_2 = 11$. From the noisy signal of noise margin, SNR of 7.39 dB.
\begin{figure}[H]
	\centering
	\includegraphics[scale = 0.3]{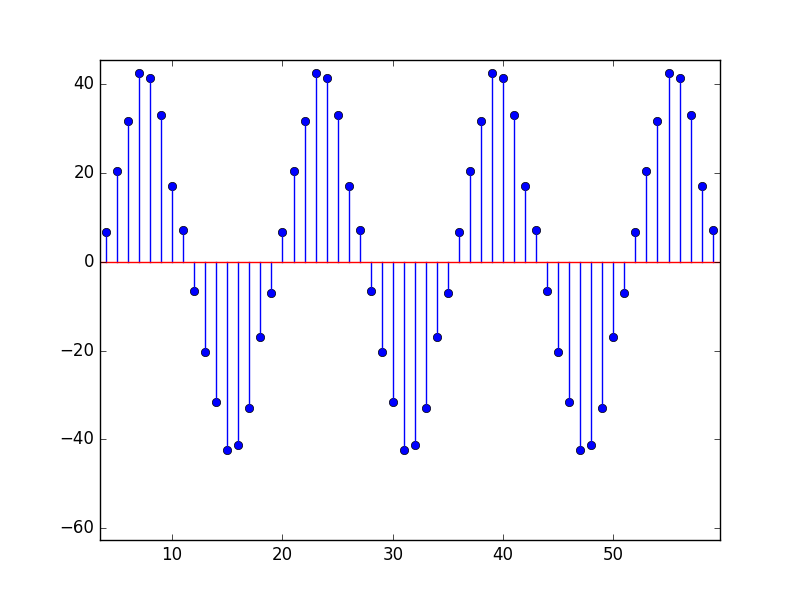}
    \caption{Reconstruction of periodic component 16}
\end{figure}

This is the $3rd$ reconstructed periodic component $p_3 = 16$. From the noisy signal of noise margin, SNR of 7.39 dB. \\
The 'resends' value in algorithm depends on the noise distribution.
\section{Conclusion}
We proposed a method to find the composite period length of an one dimensional signal in $O(n)$ order. The suggested method  is working very well till the SNR of 5dB. The biggest advantage of using this randomized process, is less computational time. If we have provided with the data length in the order of millions, SVD process will get stuck while running in the present available computational resources as the time complexity is in the order of $O(m^2n)$ where $m<n$ for a data matrix of dimension $m\times n$. That plot is shown in the result section. Once we get the composite period length, we are truncating the signal at the multiple of the composite period length to find the hidden components. Once we get the truncated signal, we are projecting it in the factor ramanujan subspaces of the composite period. Hence, exact hidden components are being able to be found. If the data length is very high, these steps are found to work extremely good. From these information, the shape or form of the individual hidden components can be approximated even in the presence of noise. \\
Another contribution of this study is to identify the DC component strength of the individual hidden periods such that the auto-correlation between the reconstructed and the actual hidden component is maximum.
\bibliographystyle{IEEEtran}
\bibliography{ref}

\begin{thebibliography}{10}
\providecommand{\url}[1]{#1}
\csname url@samestyle\endcsname
\providecommand{\newblock}{\relax}
\providecommand{\bibinfo}[2]{#2}
\providecommand{\BIBentrySTDinterwordspacing}{\spaceskip=0pt\relax}
\providecommand{\BIBentryALTinterwordstretchfactor}{4}
\providecommand{\BIBentryALTinterwordspacing}{\spaceskip=\fontdimen2\font plus
\BIBentryALTinterwordstretchfactor\fontdimen3\font minus
  \fontdimen4\font\relax}
\providecommand{\BIBforeignlanguage}[2]{{%
\expandafter\ifx\csname l@#1\endcsname\relax
\typeout{** WARNING: IEEEtran.bst: No hyphenation pattern has been}%
\typeout{** loaded for the language `#1'. Using the pattern for}%
\typeout{** the default language instead.}%
\else
\language=\csname l@#1\endcsname
\fi
#2}}
\providecommand{\BIBdecl}{\relax}
\BIBdecl

\bibitem{Ramanujan1}
P.~P. Vaidyanathan, ``Ramanujan sums in the context of signal processing—part
  i: Fundamentals,'' \emph{IEEE Transactions on Signal Processing}, vol.~62,
  no.~16, pp. 4145--4157, Aug 2014.

\bibitem{1}
S.~Tenneti and P.~P. Vaidyanathan, ``Dictionary approaches for identifying
  periodicities in data,'' in \emph{2014 48th Asilomar Conference on Signals,
  Systems and Computers}, Nov 2014, pp. 1967--1971.

\bibitem{2}
W.~A. Sethares and T.~W. Staley, ``Periodicity transforms,'' \emph{IEEE
  Transactions on Signal Processing}, vol.~47, Nov 1999.

\bibitem{Ramanujan2}
P.~P. Vaidyanathan, ``Ramanujan sums in the context of signal processing—part
  ii: Fir representations and applications,'' \emph{IEEE Transactions on Signal
  Processing}, vol.~62, no.~16, pp. 4158--4172, Aug 2014.

\bibitem{Farey}
P.~P. Vaidyanathan and P.~Pal, ``The farey-dictionary for sparse representation
  of periodic signals,'' in \emph{2014 IEEE International Conference on
  Acoustics, Speech and Signal Processing (ICASSP)}, May 2014, pp. 360--364.

\bibitem{Nested}
S.~V. Tenneti and P.~P. Vaidyanathan, ``Nested periodic matrices and
  dictionaries: New signal representations for period estimation,'' \emph{IEEE
  Transactions on Signal Processing}, vol.~63, no.~14, pp. 3736--3750, July
  2015.

\bibitem{RfilterProp}
P.~P. Vaidyanathan and S.~Tenneti, ``Properties of ramanujan filter banks,'' in
  \emph{2015 23rd European Signal Processing Conference (EUSIPCO)}, Aug 2015,
  pp. 2816--2820.

\bibitem{Rfilterbank}
S.~V. Tenneti and P.~P. Vaidyanathan, ``Ramanujan filter banks for estimation
  and tracking of periodicities,'' in \emph{2015 IEEE International Conference
  on Acoustics, Speech and Signal Processing (ICASSP)}, April 2015, pp.
  3851--3855.

\bibitem{Mindictionary}
S.~V. Tenneti and P.~P. Vaidyanathan, ``Minimal dictionaries for spanning periodic signals,'' in \emph{2015
  49th Asilomar Conference on Signals, Systems and Computers}, Nov 2015, pp.
  523--527.

\bibitem{Mindatalen}
S.~V. Tenneti and P.~P. Vaidyanathan, ``Minimum data length for integer period estimation,'' \emph{IEEE
  Transactions on Signal Processing}, vol.~66, no.~10, pp. 2733--2745, May
  2018.

\bibitem{SVD}
P.~Kanjilal and S.~Palit, ``Minimum data length for integer period
  estimation,'' \emph{IEEE Transactions on Signal Processing}, vol.~43, no.~10,
  pp. 1536--1540, June 1995.

\end{thebibliography}
\end{document}